%% file: no_go_theoremv3.tex
\documentclass[10pt,aps,pra,twocolumn,nofootinbib]{revtex4-2}
\usepackage{amsfonts,amssymb,amsmath}            
\usepackage{lmodern}
\usepackage[]{graphics,graphicx}            
\usepackage{amsthm}
\usepackage{enumerate,mathtools}
\usepackage{tikz}
\usepackage{adjustbox}
\usepackage{mathtools}
\usetikzlibrary{decorations.pathmorphing,patterns,decorations.markings,matrix,quantikz}
    \usepackage{hyperref}
\newtheorem{theorem}{Theorem}

\newtheorem{definition}{Definition}
\newtheorem{lemma}{Lemma}
\newtheorem{corollary}{Corollary}
\newcommand{\half}{\mbox{$\textstyle \frac{1}{2}$}}

\usepackage{xstring}

 \DeclareFontFamily{U}{bbold}{}
\DeclareFontShape{U}{bbold}{m}{n}
 {
  <-5.5> s*[1.069] bbold5
  <5.5-6.5> s*[1.069] bbold6
  <6.5-7.5> s*[1.069] bbold7
  <7.5-8.5> s*[1.069] bbold8
  <8.5-9.5> s*[1.069] bbold9
  <9.5-11> s*[1.069] bbold10
  <11-15> s*[1.069] bbold12
  <15-> s*[1.069] bbold17
 }{}

\DeclareRobustCommand{\identity}{%
  \text{\usefont{U}{bbold}{m}{n}1}%
}

\renewcommand{\epsilon}{\varepsilon}

\pgfset{
  foreach/parallel foreach/.style args={#1in#2via#3}{evaluate=#3 as #1 using {{#2}[#3-1]}},
}
\newcounter{arraycard}
\newcommand{\drawchain}[2][0.4\textwidth]{
\def\firstlist{#2,0}
\begin{center}
\begin{adjustbox}{width=#1}
\begin{tikzpicture}
\setcounter{arraycard}{0}
  \foreach \x in \firstlist {%
    \stepcounter{arraycard}
  }
  \foreach \x [count=\c]  in \firstlist
  {
  	\ifthenelse{\c=\value{arraycard}}{}{\draw [thick] (2*\c,0) -- node[above,pos=0.5] {\x} (2*\c+2,0);}
    \node[circle,style={fill=black,minimum width=0.8cm,text=white}] at (2*\c,0) {};
  }

\end{tikzpicture}
\end{adjustbox}
\end{center}
}
\newcommand{\drawchainwithfields}[3][0.4\textwidth]{
\def\firstlist{#2,0}
\def\secondlist{#3}
\begin{center}
\begin{adjustbox}{width=#1}
\begin{tikzpicture}
\setcounter{arraycard}{0}
  \foreach \x in \firstlist {%
    \stepcounter{arraycard}
  }
  \foreach \x [count=\c,
  parallel foreach=\y in \secondlist via \c]
  in \firstlist
  {
  \IfInteger{\y}{
    \ifthenelse{\y=0}{}{
    \node[anchor=south] at (2*\c,0.4cm) {\scriptsize\y};
    }
    }{
	\node[anchor=south] at (2*\c,0.4cm) {\scriptsize\y};
    }
    \ifthenelse{\c=\value{arraycard}}{}{\draw [thick] (2*\c,0) -- node[above,pos=0.5] {\x} (2*\c+2,0);}
    \node[circle,style={fill=black,minimum width=0.8cm,text=white}] at (2*\c,0) {};
  }

\end{tikzpicture}
\end{adjustbox}
\end{center}
}

\begin{document}

\title{The Limits of Quantum State Transfer for Field-Free Heisenberg Chains}
\date{\today}

\author{Alastair \surname{Kay}}
\affiliation{Department of Mathematics, Royal Holloway University of London, Egham, Surrey, TW20 0EX, UK}
\email{alastair.kay@rhul.ac.uk}
\begin{abstract}
In a one-dimensional Heisenberg chain, we show that there are no sets of coupling strengths such that the evolution perfectly transfers a quantum state between the two ends of the chain without the addition of magnetic fields. In lieu of perfect transfer, we consider a range of options for achieving high quality transfer, whether in finite time, or via ``pretty good'' transfer where one waits long times in the hope of getting arbitrarily close to perfect transfer. In attempting to engineer arbitrarily accurate transfer, we explore a new paradigm that facilitates time estimates for achieving any target accuracy $\epsilon$ for the transfer.
\end{abstract}

\maketitle

\section{Introduction}

True mastery of a physical theory is demonstrated when we transition from observing physical phenomena to explaining them, and ultimately to controlling them, seeking to induce particular behaviour for a practical purpose. This is the aim for the burgeoning field of quantum technologies. The myriad challenges, with decoherence taking centre stage, are apparent in the fact that we still remain some distance from a fully operational, universal, scalable quantum computer. Many quantum technologists are thus focussed on simpler, short-term applications with more limited scope. Nevertheless, achieving control of multiple quantum bits and maintaining coherence remains a challenge.

Reduced complexity of control sequences, and implementation time, can have a huge impact on the practicality of any given protocol. In scenarios such as the transfer of a quantum state \cite{bose2003,christandl2004,kay2010a}, the generation of GHZ states \cite{clark2005,kay2017c} and optimal cloning \cite{kay2017c}, it is possible to almost entirely dispense with time varying control of a system by relying upon the evolution of a carefully tuned, fixed, Hamiltonian. These protocols even demonstrate a reduction in implementation time, and hence decoherence, compared to the traditional gate model. The underlying theory is broadly applicable, including the Heisenberg and exchange models, which in turn translate to a wide range of experimental scenarios including the solid state \cite{majer2007,plantenberg2007}, trapped ions \cite{islam2011}, or even photonic systems \cite{perez-leija2013,chapman2016}.

What are the ultimate limits of these restrictions? We will primarily focus on state transfer, as it is the best understood, but other state synthesis tasks \cite{kay2017,kay2017a,kay2017c} could be considered. Perfect transfer for uniformly coupled systems is impossible in all but the shortest chains \cite{bose2003,christandl2005}. Magnetic fields can be added that enhance the quality of state transfer \cite{shi2005}, but perfect transfer remains impossible \cite{kempton2017} without engineering the coupling strengths.

In this paper, we take the counterpoint of that approach, using engineered couplings, but `field-free', i.e.\ with no (or uniform) magnetic fields. This would further reduce the experimental control required in synthesising the Hamiltonian. Such a restriction creates a distinction between different Hamiltonian models, such as the exchange and Heisenberg models. In Section \ref{sec:imposible}, we shall see that although the field-free restriction is largely irrelevant to the exchange Hamiltonian, perfect state transfer is impossible with a field-free Heisenberg model. 

In the absence of perfect transfer, we consider the potential for high quality state transfer in Sec.\ \ref{sec:spec_constrain} and give estimates on the minimum time to achieve this, drawing negative comparisons with the field-full case. Sec.\ \ref{sec:pretty} considers another possibility that remains: arbitrarily accurate transfer. One possible definition of this, ``pretty good transfer'', has been studied for uniformly coupled systems of both the Heisenberg and exchange variety \cite{coutinho2016,vanbommel2016,banchi2017}. However, no time estimates have been forthcoming. Here, we define a new paradigm that is appropriate to engineered systems. Although our best time estimates are comparable to the age of the Universe, and consistent with extrapolations from previous studies of perfect recurrences \cite{hemmer1958,peres1982,bhattacharyya1986}, there are good prospects for improvement; we discuss some options in Sec.\ \ref{sec:future} and assess the potential that they engender.

\subsection{Perfect State Transfer}

We start by reviewing the context and requirements for perfect state transfer. For excitation preserving Hamiltonians, i.e.\ where the Hamiltonian satisfies
$$
\left[H,\sum_{n=1}^NZ_n\right]=0,
$$
and $Z_n$ is the Pauli-$Z$ matrix applied to qubit $n$ of $N$, it is sufficient for us to focus on the problem of excitation transfer for a single excitation, which must evolve within the single excitation subspace, i.e.\ the space in which there is always exactly one qubit in state $\ket{1}$, and the others in $\ket{0}$. If an evolution of the form
$$
e^{-iHt_0}\ket{1000\ldots 0}=e^{i\phi}\ket{00\ldots 01}
$$
can be achieved for some phase $\phi$, then when an unknown state $\ket{\psi}$ is placed on the first qubit of a chain that is otherwise in the state $\ket{0}^{\otimes(N-1)}$, that state is perfectly transferred to the final spin, up to a corrective phase gate. The time $t_0$ is known as the state transfer time.

Two standard Hamiltonian forms, Heisenberg and exchange, model a wide variety of experimental systems. 
\begin{align}
H_{\text{Ex}}=&\frac12\sum_{n=1}^{N-1}J_n(X_nX_{n+1}+Y_nY_{n+1})+\frac12\sum_{n=1}^NB_nZ_n\label{eqn:Exmodel} \\
H_{\text{Heis}}=&H_{\text{Ex}}+\frac12\sum_{n=1}^{N-1}J_nZ_nZ_{n+1} \label{eqn:Heismodel}
\end{align}
(Here we consider a one-dimensional geometry.) The two are equivalent from the perspective of the single excitation subspace, both mapping to an arbitrary real, symmetric, tridiagonal matrix. For that reason, the literature often hasn't needed to distinguish between which underlying model is chosen: any difference can be absorbed into the $\{B_n\}$. One exception is studies of transfer in uniformly coupled systems, $J_n=1, B_n=0$ \cite{bose2003,banchi2017,coutinho2016,vanbommel2016}, where choice of model is crucial\footnote{Another notable exception is when multiple excitations are considered, where analysis of the exchange Hamiltonian is facilitated by the Jordan-Wigner transformation.}.

Requiring a field-free Hamiltonian, in which $\{B_n\}=0$, similarly separates the two models. For the exchange model, this imposes that eigenvalues must occur in $\pm\lambda$ pairs (with a 0 if the chain length is odd), while for the Heisenberg model, one of the eigenvalues must be 0, and the null vector must be the uniform superposition of all sites. How do these restrictions impact upon perfect and pretty good state transfer protocols?

When the single excitation subspace can be described by a real symmetric tridiagonal $N\times N$ matrix, the conditions for perfect state transfer within the 0/1 excitation subspaces are well understood \cite{kay2010a}:

\begin{lemma}\label{lem:PST}
End-to-end perfect state transfer can be achieved in a nearest-neighbour coupled chain with positive couplings $J_n>0$ in time $t_0$ if and only if (i) the system is centrosymmetric ($J_n=J_{N-n}$ and $B_n=B_{N+1-n}$); and (ii) the ordered eigenvalues $\{\lambda_n\}$ satisfy $(\lambda_n-\lambda_{n+1})t_0/\pi\in2\mathbb{N}-1$.
\end{lemma}

Select any spectrum that satisfies those properties and you can reverse engineer the couplings that yield that spectrum via an inverse eigenvalue problem \cite{karbach2005,gladwell2005}. Thus, it is simple to specify appropriate spectra and design corresponding chains. Indeed, the standard choice \cite{christandl2004} uses a spectrum $0,\pm1,\pm2,\pm3,\ldots$ and is consequently field-free for the exchange model.

Since all perfect state transfer systems have a perfect revival, wherein the excitation reappears perfectly on the input site, at time $2t_0$, it is often helpful to consider the conditions for a perfect revival/recurrence.
\begin{lemma}\label{lem:Previve}
The first site on a nearest-neighbour coupled chain with positive couplings $J_n>0$ exhibits a perfect revival in time $t_r$ if and only if the ordered eigenvalues $\{\lambda_n\}$ satisfy $(\lambda_n-\lambda_{n+1})t_r/\pi\in2\mathbb{N}$.
\end{lemma}

\subsection{Pretty Good Transfer}

In most scenarios where coupling strengths can be engineered, perfect transfer is possible, and the only discussion that can remain is whether high quality transfer can be realised at shorter times. However, we will see that perfect transfer for the field-free Heisenberg model is impossible, which leaves a broad spectrum of possibilities. Is high quality transfer possible? How long does it take? A limiting case is arbitrarily accurate state transfer. This has been formalised as `pretty good transfer', wherein it is required that for any error $\epsilon>0$, there should exist a time $t_\epsilon>0$ such that the fidelity satisfies
$$
F:=|\bra{00\ldots 01}e^{-iHt_\epsilon}\ket{1000\ldots 0}|^2>1-\epsilon.
$$
A characterisation of pretty good transfer was given in \cite{banchi2017}, and is here adapted to the specific scenario:

\begin{lemma}\label{lem:PGST}
End-to-end pretty good state transfer can be achieved in a nearest-neighbour coupled chain with positive couplings $J_n>0$ if and only if (i) the matrix describing the single excitation subspace is centrosymmetric; and (ii) for all sets of integers $\{l_i\}$ such that $\sum_il_i\lambda_i=0$ and $\sum_{i}l_{2i}$ is odd, $\sum_il_i\neq 0$.
\end{lemma}

For the field-free Heisenberg model, the second condition simplifies to the requirement that if $\sum_il_i\lambda_i=0$, $\sum_il_{2i}$ must be even because $\lambda_1=0$, so $l_1$ can be chosen arbitrarily. We also note that the analysis here is closely related to the idea of perfect revivals (we shall make this connection more explicit later). In the case where there are no sets of integers $\{l_i\}$ such that $\sum_il_i\lambda_i=0$, there are useful estimates on the times at which perfect revivals occur \cite{peres1982,bhattacharyya1986,hemmer1958}. However, the point that these estimates make is that even for modest sized systems, the recurrence time is longer than the age of the Universe, and that this is not a useful phenomenon.

\section{The Impossibility of Perfect Transfer}\label{sec:imposible}

We will now prove our main claim -- that with the exception of $N=2$, there are no end-to-end perfect transfer chains for the field-free Heisenberg model. Our proof strategy is reminiscent of one used in \cite{kay2018}. First, we observe that if $h$ is the restriction of $H$ on the first excitation subspace, then in the field-free case it must satisfy
$$
h\sum_{n=1}^{N}\ket{n}=0,
$$
and $h$ is non-positive. Here we are using $\{\ket{n}\}$ as a basis of the $N$-dimensional space. One can think of $\ket{n}$ as specifying that there is a $\ket{1}$ on qubit $n$, and $\ket{0}$ on all other qubits. This tells us that the matrix must have a specific null vector $\ket{\lambda_1}$ as well as the desired spectrum. This is closely connected with recent studies of quantum state synthesis \cite{kay2017} which also tried to fix a spectrum and a null vector. In particular, this imposes
$$
\braket{\lambda_1}{1}^2=\frac{1}{N}.
$$
It is this which we shall show is impossible by virtue of the fact that the denominator contains no more than $\lfloor\log_2(N)\rfloor$ factors of two. Since the chain is centrosymmetric, as required by Lemma \ref{lem:PST}, if we know the eigenvalues, we can write down the first elements of the eigenvectors. In particular,
\begin{equation}
\braket{1}{\lambda_n}^2=R\frac{(-1)^{n+1}}{\displaystyle\prod_{\substack{m=1\\m\neq n}}^{N}\!(\lambda_n-\lambda_m)}\label{eqn:evecels}
\end{equation}
for some constant of proportionality $R$ such that $\sum_n\braket{1}{\lambda_n}^2=1$ \cite{gladwell2005}. Without loss of generality, we can take the $\lambda_n$ to be integers $p_n$ of alternating parity. Any constant of proportionality would simply be incorporated into the $R$ and hence cancelled.

\begin{lemma}\label{lem:Rdenom}
In a chain capable of perfect end-to-end state transfer, $R$ must be rational and its irreducible form must contain at least one factor of 2 in the denominator.
\end{lemma}
\begin{proof}
This has previously been proven in Lemma 4 of \cite{kay2018}. However, we shall now give an alternative proof.

We henceforth focus on writing $R$ with a common denominator, and determining how many factors of two are in the numerator. By symmetry $\braket{1}{\lambda_n}=(-1)^{n+1}\braket{N}{\lambda_n}$ and $\sum_n\braket{1}{\lambda_n}\braket{\lambda_n}{N}=\bra{1}H\ket{N}=0$. We can thus subtract the two conditions, leaving
$$
2R\sum_{n=1}^{\lfloor N/2\rfloor} \frac{1}{\prod_{m=1,m\neq 2n}^{N}(\lambda_{2n}-\lambda_m)}=1.
$$
We put this all over a common denominator, which would start as being $\prod_{n,m:2n\neq m}(\lambda_{2n}-\lambda_m)$. However, there will be many factors in the numerator that will cancel. We're focussed on counting factors of two, which means we're only concerned with even terms, which only arise as $(\lambda_{2p}-\lambda_{2q})$. Every term in the sum over $n$ contains the factor $(\lambda_{2p}-\lambda_{2q})$ except for two: $n=p,q$. However, consider these terms which, up to a common factor, are
\begin{equation}
\lambda_{2q}\!\!\prod_{\substack{m=2\\n\neq 2q,2p}}^N\!\!(\lambda_{2q}-\lambda_m)+(-1)^{2q-2p-1}\lambda_{2p}\!\!\prod_{\substack{m=2\\m\neq 2p,2q}}^N\!\!(\lambda_{2p}-\lambda_m).\label{eqn:highpowers}
\end{equation}
Since $2q-2p-1$ is always odd, Eq.\ (\ref{eqn:highpowers}) is an odd function of $(\lambda_{2q}-\lambda_{2p})$. Hence, the entire numerator has this as a factor, which can be cancelled from the denominator. Repeating for all such factors cancels all even terms from the denominator of $1/2R$. $R$ is a rational number with an even denominator in the irreducible form.
\end{proof}

\begin{theorem}\label{thm:notransfer}
No field-free Heisenberg chains of length $N\geq 3$ are capable of perfect end-to-end state transfer.
\end{theorem}
\begin{proof}
For a Heisenberg chain with positive coupling strengths, the largest eigenvalue is $\lambda_1=0$. Thus,
$$
\braket{1}{\lambda_1}^2=R\frac{(-1)^{N-1}}{\displaystyle\prod_{m=2}^{N}\!\lambda_m}.
$$
As before, it is sufficient to consider the $\lambda_m$ being integers, such that $R$ is rational with a factor of 2 in the denominator. Moreover, the values $\lambda_{2m}$ are odd and $\lambda_{2m-1}$ are even. Hence the product $\displaystyle\prod_{m=2}^{N}\!\lambda_m$ contains at least $\lfloor (N-1)/2\rfloor$ factors of 2. Thus, any chain that is capable of perfect end-to-end transfer must have a rational value of $\braket{\lambda_1}{1}^2$, with the denominator of the irreducible form containing at least $\lfloor (N+1)/2\rfloor$ powers of two. Since the field-free Heisenberg chain has $\braket{\lambda_1}{1}^2=\frac{1}{N}$, this contains at most $\lfloor\log_2N\rfloor$ factors of 2.

The only instances in which we can reconcile $\lfloor\log_2N\rfloor\geq\lfloor (N+1)/2\rfloor$ are $N=2,4$. The first of these is well known \cite{bose2003}. We eliminate $N=4$ by repeating the proof and explicitly setting the spectrum to $(0,2a+1,2b,2c+1)$. This fixes
$$
\braket{1}{\lambda_1}^2=-\frac{(2a+1-2b)(2c+1-2b)}{8b(a+c+1-b)}.
$$
Since the denominator contains a factor of 8 and the numerator is odd, this cannot be $\braket{1}{\lambda_1}^2\neq\frac14$. We conclude that, no matter how you choose the coupling strengths of a chain, the only field-free Heisenberg model to exhibit end-to-end perfect state transfer is the $N=2$ case.

\end{proof}

These results are consistent with previous results \cite{coutinho2014,alvir2016} on uniform weighted graphs with the Heisenberg model (for which the Hamiltonian in the single excitation subspace is just the Laplacian of the graph).

\section{Spectral Constraints on Field-Free Heisenberg Models}\label{sec:spec_constrain}

In the absence of perfect state transfer, we are now on a mission to see what properties we can recover. Are perfect revivals still possible? Is high fidelity transfer possible? What is the minimum required time (as a function of chain length) to achieve a high transfer fidelity?

Our understanding of state transfer is intimately linked with the spectral properties of the chain. As such, we wish to understand some constraints on that spectrum.

\begin{lemma}\label{lem:evalbound}
If the maximum coupling strength of a field-free Heisenberg model is constrained to $J_{\max}$, then the ordered eigenvalues $\lambda_n$ are bound by
$$
\lambda_n\geq 2J_{\max}\left(\cos\left(\frac{\pi(n-1)}{N}\right)-1\right).
$$
\end{lemma}
\begin{proof}
Let $E=\sum_{i=1}^{N-1}\proj{i}-\sum_{i=1}^{N-1}\ket{i+1}\bra{i}$, $J_i=J_{\max}\tilde J_i$ and $\tilde J=\sum_{i=1}^{N-1}\tilde J_i\proj{i}$. The matrix $h=-J_{\max}E\tilde JE^T$ is similar to $\tilde h=-J_{\max}\tilde JE^TE$ (two matrices $AB$ and $BA$ are similar, so let $A=E$ and $B=\tilde JE^T$). For any two matrices $A,B$, their ordered singular values satisfy 
$\sigma_i(AB)\leq\sigma_i(A)\|B\|$ \cite{hogben2014}.
The singular values are the absolute value of the eigenvalues, which are negative for $h$ and $\tilde h$.
Since $\|\tilde J\|= 1$, it must be that
$$
\lambda_n\geq J_{\max}\eta_n
$$
where $\eta_n$ are the eigenvalues of $-E^TE$ \cite{bose2003}:
$$
\eta_n=2\cos\left(\frac{\pi(n-1)}{N}\right)-2.
$$
\end{proof}

The lemma reveals some crucial properties. For example, the smallest eigenvalue gap can be no larger than $2J_{\max}(1-\cos\frac{\pi}{N})\sim\frac{\pi^2J_{\max}}{N^2}$. Had perfect transfer been possible, this indicates that the transfer time scales as $\sim N^2/J_{\max}$. This is $O(N)$ worse than the fastest transfer that was possible with the exchange model \cite{yung2006,kay2016b}. It also shows that in models where perfect revival is possible, the perfect revival time must scale in the same way.

Moreover, even in the case of imperfect state transfer, this result strongly suggests that it will require a time $\Omega(N^2/J_{\max})$ in order to achieve high fidelity. To see this, recall Equation (\ref{eqn:evecels}) -- a large weight for $|\braket{1}{\lambda_n}|^2$ is achieved for those eigenvectors with the closest eigenvalue spacings. Since one needs the eigenvalue conditions of perfect state transfer to at least be approximately satisfied on a subset of eigenvalues whose total weight is large in order to get high quality transfer, it is strongly indicated that the relevant gaps will be $O(J_{\max}/N^2)$. Again, since we can already perform vastly better with the introduction of a magnetic field, the use of the exchange model, or a very minimal amount of control \cite{murphy2010}, these use cases are largely eliminated. Nevertheless, for completeness, we wish to enumerate some possibilities.

\subsection{Perfect Revivals in a Field-Free Setting}\label{sec:revival}

Since perfect state transfer models all have a perfect revival, it is natural to wonder whether any field-free Heisenberg model can exhibit perfect revivals. We have already seen that the minimal time for such a revival is $J_{\max}t_0\sim N^2$. We now report a set of couplings that saturates this scaling. This model was first stated in \cite{albanese2004}, although was not recognised as a field-free Heisenberg model, and is based on the Hahn polynomials:
\begin{equation}
J_n=n(N-n). \label{eqn:PGcouplings}
\end{equation}
The spectrum is $\{-n(n-1)\}_{n=1}^{N}$. The eigenvectors $\ket{\lambda_n}$ are given in \cite{albanese2004}, including the final elements:
\begin{equation}
|\braket{\lambda_n}{N}|^2=\frac{(2n-1)(N-1)!^2}{(N+n-1)!(N-n)!}=\frac{\binom{2N-2}{N-n}-\binom{2N-2}{N-n-1}}{\binom{2N-2}{N-1}}. \label{eqn:start}
\end{equation}
Since every eigenvalue is an even integer, this chain has a perfect revival at time $t_0=\pi$, with 0 phase (i.e.\ $J_{\max}t_0=N^2\pi/4$, which is only a factor of $\pi^2/8$ longer than the predicted limit). While the matrix is centrosymmetric, its eigenvalues are not compatible with perfect transfer at time $\pi/2$ for $N>2$. 

\subsection{Numerical Solutions for High Fidelity Transfer}\label{sec:numbers}

We now wish to construct some high transfer fidelity examples of Heisenberg chains, where the high fidelity is achieved at close to the optimum time (rather than waiting arbitrarily long times, as in Sec.\ \ref{sec:pretty}). To do this, we note that the similarity transforms of Lemma \ref{lem:evalbound} hint at a very close connection to the uniformly coupled case. As such, we use this as our starting point and recall the numerical perturbative strategy of \cite{karbach2005}: we fix the smallest eigenvalue gap, $\delta$, and try to perturb the system such that subsequent gaps are integer multiples of $\delta$. Since this would yield perfect transfer at $t_0=\pi/\delta$ (which is $O(N^2)$), Theorem \ref{thm:notransfer} proves that it is impossible to achieve with \emph{all} eigenvalues. Parameter counting suggests that it should be possible with approximately $N/2$ of the eigenvalues -- we have $\lfloor N/2\rfloor$ coupling parameters in a symmetric system, suggesting we can control about this many eigenvalues. To maximise the fidelity, we choose to fix the eigenvalues with the largest weights of the eigenvectors on the first/last site.

Given the perturbative strategy, the eigenvector elements are close to those of the unperturbed system \cite{bose2003},
$$
|\braket{1}{\lambda_n}|^2=\frac{2-\delta_{n,1}}{N}\cos^2\left(\frac{\pi (n-1)}{2N}\right).
$$
The total weight of the `good' eigenvectors is then
$$
W=\sum_{n=1}^{N/2}\frac{2-\delta_{n,1}}{N}\cos^2\frac{\pi (n-1)}{2N}=\frac{N-1}{2N}+\frac{1}{2N\tan\frac{\pi}{2N}}.
$$
In the large $N$ limit, we have $W\rightarrow \frac{1}{2}+\frac{1}{\pi}\approx 0.82$.
A typical excitation fidelity can be expected to be about $F_\text{ex}=W$ (all the `good' eigenvalues aligning perfectly, and the unfixed ones randomly distributed), yielding a state transfer fidelity of $F=\frac13+(1+\sqrt{F_\text{ex}})^2/6\approx 0.94$ while the excitation transfer fidelity should be no worse than $F_\text{ex}=2W-1$ (all the `good' eigenvalues aligning perfectly, and the others destructively interfering).

\begin{figure}
\centering
\includegraphics[width=0.45\textwidth]{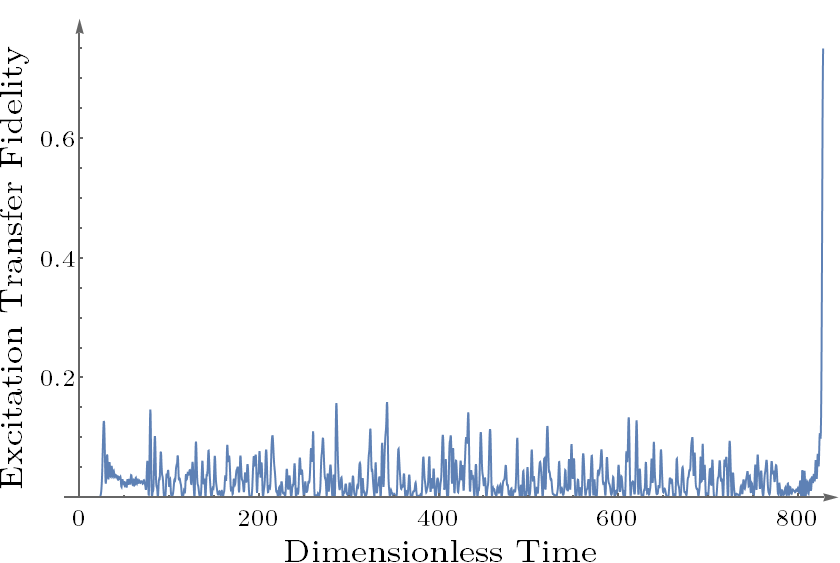}
\caption{Time evolution of a field-free chain of 51 qubits, modified from a uniform coupling to ensure high quality transfer.}\label{fig:example}
\end{figure}

In terms of the perturbative strategy we require, we can introduce a multiplicative change
\begin{equation}\label{eqn:perturb}
J_{\max}\tilde J^{(n+1)}2=J_{\max}\tilde J^{(n)}(\identity+\delta J).
\end{equation}
The new eigenvalues $\lambda^{(n+1)}_m$ can be written as an update from the previous values up to $O(\delta J^2)$,
\begin{equation}\label{eqn:evalperturb}
\frac{\lambda^{(n+1)}_m-\lambda^{(n)}_m}{\lambda^{(n)}_m}=\bra{\lambda^{(n)}_m}\delta J\ket{\lambda^{(n)}_m}.
\end{equation}
Note that these eigenvectors $\ket{\lambda^{(n)}_n}$ are the eigenvectors of $\tilde JE^TE$, \emph{not} our original matrix $E\tilde JE^T$. This is a linear problem for the change in coupling strengths $\delta J$ that we can solve and iterate towards improved values. In practice, we achieve convergence to machine precision to the target parameters within about 5 steps. For instance, we created a chain of 51 qubits with an excitation transfer fidelity of $F_{\text{ex}}=0.749$ by fixing the 25 largest non-zero eigenvalues. The evolution is depicted in Fig.\ \ref{fig:example}.

\section{Pretty Good State Transfer}\label{sec:pretty}

The impossibility of perfect transfer in a field-free setting conveys nothing about what can be achieved in the case of pretty good state transfer, where one assesses whether arbitrarily accurate transfer can be achieved by waiting long enough. This is because the eigenvalues are no longer constrained to have the perfect integer spacing. Indeed, pretty good transfer has been demonstrated in some cases of uniform field-free Heisenberg models \cite{banchi2017}. To date, these analyses lack an estimate of the time $t_\epsilon$ to achieve a particular maximum error $\epsilon$.

We might approach a rough estimate in the following way. Firstly, we note that a pretty good transfer is achieved by finding a $t_\epsilon$ and integers $p_n$ such that
$$
|\lambda_nt_\epsilon/\pi-p_n-\gamma|<\epsilon
$$
where $e^{-ip_n\pi}=(-1)^{n+1}$. Clearly, this means that 
$$
|2\lambda_nt_\epsilon/\pi-2p_n-2\gamma|<2\epsilon.
$$
In other words, a pretty good transfer of inaccuracy $\epsilon$ in time $t_\epsilon$ will yield a pretty good revival (inaccuracy $2\epsilon$) in time $2t_\epsilon$. Since the converse need not be true, estimating the length of time for perfect revivals gives a lower bound on the length of time for perfect transfer. There have been a number of such attempts \cite{peres1982,bhattacharyya1986,hemmer1958}, but these all assume that the eigenvalues are independent over the rationals, which is not the case (see, e.g.\ \cite{banchi2017}). They all estimate a time of order $\epsilon^{-N}$ for system size $N$ and accuracy $\epsilon$, and conclude that even for very modest systems, this time is longer than the age of the Universe. We conjecture that a good estimate on the time can be determined by replacing the system size $N$ in $\epsilon^{-N}$ with the number of independent eigenvalues. This suggests that it would take highly exceptional cases to permit a reasonable protocol based on pretty good transfer.

Can we construct pretty good transfer instances for an engineered field-free Heisenberg chain? Obviously, we would like to make the time $t_\epsilon$ as small as possible. One promising strategy is to recognise that our chain, which we now fix to be of length $2N$, must possess mirror symmetry (must be centrosymmetric).
\drawchain[0.48\textwidth]{$J_1$,$J_2$,$J_3$,$J_4$,$a$,$J_4$,$J_3$,$J_2$,$J_1$}
Hence, we can decompose the chain into symmetric and antisymmetric subspaces, each being an effective chain of length $N$. The symmetric subspace can be depicted as
\drawchain[0.32\textwidth]{$J_1$,$J_2$,$J_3$,$J_4$}
where each link corresponds to a Heisenberg coupling of the indicated strength. The antisymmetric subspace is
\drawchainwithfields[0.32\textwidth]{$J_1$,$J_2$,$J_3$,$J_4$}{0,0,0,0,"$2a$"}
where the $2a$ term acts as an additional effective magnetic field. In this picture, perfect state transfer is just the process of transferring
$$
(\ket{1}+\ket{2N})+(\ket{1}-\ket{2N})\rightarrow (\ket{1}+\ket{2N})-(\ket{1}-\ket{2N}).
$$
In other words, both effective chains want perfect revivals at the same time, with a relative phase of $\pi$ between the two reviving states. Hence, if we can build the symmetric chain so that it has perfect (rather than arbitrarily accurate) revivals, that is half the challenge completed. Indeed, it halves the number of rationally independent eigenvalues, and the time estimate is essentially that of the perfect revival for just the anti-symmetric chain (reinforcing our conjecture about the role of $N$). 

For the sake of concreteness, we can use the analytic solution specified in Sec.\ \ref{sec:revival}. We still have freedom to select the parameter $a$ in order to try and encourage that, for the antisymmetric subspace: (i) a pretty good revival occurs at some integer multiple of $t_0$, and (ii) the phase of the reviving state is $-1$ relative to the symmetric space. Since pretty good revivals are generic \cite{bocchieri1957}, we should concentrate on the relative phase. Let the eigenvalues of the symmetric subspace be $\lambda_n$, and those of the antisymmetric subspace be $\lambda_n^-$. They satisfy an interlacing property $\lambda_n>\lambda_{n+1}^->\lambda_{n+1}$. The symmetric subspace has perfect revivals at all times $kt_0$ for $k\in\mathbb{Z}$. Now, let's assume that, at time $t_1=kt_0$, the antisymmetric subspace has a perfect revival, but with phase $\phi$. Since
$$
(\lambda_n^--\lambda_n)t_1=\phi+2p\pi,
$$
it follows that
$$
\sum_n\lambda_n^--\sum_n\lambda_n=\frac{N\phi+2q\pi}{t_1}.
$$
The left-hand side is just $\text{Tr}(H_+)-\text{Tr}(H_-)=2a$, where $H_\pm$ are the anti/symmetric subspaces of the Hamiltonian in the single excitation subspace. If $t_1=t_0$, $q=0$ (different integer values can be chosen here) and $\phi=\pi$, then at every second opportunity of perfect revival on the symmetric subspace, $(2k+1)t_0$, the phase on the antisymmetric subspace is $-1$. So, if its pretty good revival coincides with that time, it will work. Hence, we would set $a=N/2$. However, to date, we have not proven a single case that satisfies Lemma \ref{lem:PGST}.

\subsection{Perturbative Approach}

Instead, let us return to the basic construction of our model for pretty good transfer, with Hamiltonian $h(a)$ based on the single parameter $a$ coupling two copies of the perfect revival chain that we introduced in Eq.\ (\ref{eqn:PGcouplings}), so that we have perfect revivals on the symmetric subspace. We will take advantage of the possibility to change coupling strengths, not available when studying uniform coupling, to demonstrate a ``pretty good route'' towards arbitrarily accurate state transfer.
\begin{definition}
The Hamiltonians $h(a)$ present a pretty good route to state transfer if, for any target accuracy $\epsilon$ there exists a parameter $a_{\epsilon}$ such that $h(a_{\epsilon})$ achieves a state transfer fidelity of at least $1-\epsilon$ in a time $t_{\epsilon}$.
\end{definition}
\noindent Instead of just adjusting $t$, we have the possibility to adjust $a$ as well, making the analysis much simpler.

In particular, let us take $a\ll 1$, permitting a perturbative expansion in the antisymmetric subspace:
$$
\lambda_n^-\approx\lambda_n+2a|\braket{\lambda_n}{N}|^2+O(a^2).
$$
Provided $a$ is small enough, the $O(a^2)$ terms are negligible. At times $kt_0$ for $k\in\mathbb{N}$, we know that $\lambda_nkt_0/\pi$ is an even integer. If $2a|\braket{\lambda_n}{N}|^2kt_0/\pi$ is an odd integer, we have perfect transfer up to the accuracy of the perturbative expansion, $O(ka^2)$. Eq.\ (\ref{eqn:start}) identifies the values $|\braket{\lambda_n}{N}|^2$. Let $a$ be a very small rational number, and
$$
k=\frac{1}{2a}\binom{2N-2}{N-1}=\frac{1}{2a}|\braket{\lambda_{N}}{N}|^2.
$$
In the special case of $N=2^r$, Kummer's Theorem conveys that $\binom{2N-2}{N-1}|\braket{\lambda_n}{N}|^2$ must be odd.
Hence, $2a|\braket{\lambda_n}{N}|^2k$ is an odd integer for all $n$, as required. The only error in the state transfer process is the result of the perturbative expansion. Thus, $a\rightarrow 0$ yields a pretty good route to state transfer. The fidelity is
\begin{align*}
F&=\left|\frac12-\frac12\sum_{n=1}^N\left|\braket{1}{\tilde\lambda_n}\right|^2e^{-i\lambda_n^-kt_0}\right|^2	\\
&=\left|\frac12+\frac12\sum_{n=1}^N\left|\braket{1}{\tilde\lambda_n}\right|^2e^{-i\delta\lambda_nkt_0}\right|^2
\end{align*}
where $\ket{\tilde\lambda_n}$ are the new eigenvectors and $\delta\lambda_n\sim a^2$ is the error in the new eigenvalues not taken into account by the first order approximation. This gives that $F=1$ up to a term $O(k^2a^4)$. Hence, if we select
$$
a\sim\frac{\sqrt{\epsilon}}{ka}=\frac{\sqrt{\epsilon}}{\binom{2N-2}{N-1}},
$$
any desired accuracy $\epsilon$ can be achieved for fixed $N$. The state transfer time consequently scales as
\begin{equation}\label{eqn:time_scale_pretty}
T\sim\frac{1}{\sqrt\epsilon}\binom{2N-2}{N-1}^2\sim\frac{4^N}{N\sqrt\epsilon}.
\end{equation}
Note that, in particular, this gives an exponential improvement in dependence upon $\epsilon$ as $N$ changes compared to our rough predictions for the standard concept of pretty good transfer. Nevertheless, times are still prohibitive. Even for $N=4$, a time of $12000t_0$ is required in order to achieve an $\epsilon<10^{-5}$, while for $N=16$, we already require a time $>10^{17}t_0$. This is an inherent flaw in our chosen route, being severely impacted by such a small value of $a$ and the corresponding time necessary to acquire sufficient phase difference between the symmetric and antisymmetric parts of the chain.

\subsubsection{Scaling Improvement}

Imagine that we want to achieve an $\epsilon$ that is larger than $2|\braket{\lambda_N}{1}|^2$ (for the sake of this argument, we will neglect the difference between $|\braket{\lambda_N}{1}|^2$ and $|\braket{\tilde\lambda_N}{1}|^2$). This already allows us to achieve the scale required for fault tolerance, even at $N=8$. In this case, it is not necessary to get a phase of $-1$ from every eigenvalue. Even in the worst case where the last eigenvector gives a phase of $+1$, we can achieve the fidelity $|\half+\half(1-2|\braket{\lambda_N}{1}|^2)|^2>1-\epsilon$. In principle, this can be achieved for a smaller value of $k$ than used above, and hence the time can be shortened.

Taking this as a serious proposition, we keep only the largest $M$ eigenvalues. The total weight of these terms is
$$
\frac{\sum_{n=1}^M\binom{2N-2}{N-n}-\binom{2N-2}{N-n-1}}{\binom{2N-2}{N-1}}=1-\frac{\binom{2N-2}{N-M-1}}{\binom{2N-2}{N-1}},
$$
facilitating a fidelity
$$
F\sim1-2\frac{\binom{2N-2}{N-M-1}}{\binom{2N-2}{N-1}}>1-2\left(\frac{N-1}{N-M-1}\right)^{M}.
$$
Assuming $M\ll N$, $F\sim 1-2e^{-M^2/N}$. Thus, we are motivated to select $M\sim \sqrt{N}\log\frac{1}{\epsilon}$.

What improvement in $k$ can be expected by using only $O(\sqrt{N})$ eigenvectors instead of all $N$? The best time is readily calculated for a particular case:
\begin{equation}
k=\frac{1}{2a}\text{gcd}\left\{\binom{2N-2}{N-n}-\binom{2N-2}{N-n-1}\right\}_{n=1}^M. \label{eqn:k}
\end{equation}
Again, for $N=2^r$, we are already guaranteed that all the integer values will be odd (it may be that by neglecting some values, other $N$ also become a possibility). To proceed analytically, we assume $k=\frac{1}{2a}\frac{(N+M)!}{(N-1)!}$, as this is certainly sufficient to remove the denominators from all the $\{2a|\braket{\lambda_n}{1}|^2\}_{n=1}^M$, although it may not be the optimum value. This means that the time scales like
$$
T\sim (N+M)^{2M}\sim e^{\sqrt{N}\log\frac1{\epsilon}\log N+\log^2\frac1{\epsilon}},
$$
which is a marked improvement over Eq.\ \ref{eqn:time_scale_pretty} (but cannot be used to arbitrary $\epsilon$). Nevertheless, the times remain prohibitive for all but the shortest chains.


\subsection{Future Prospects}\label{sec:future}

The pressing question for the future is whether some development of the current methodology could present better run-times. Obviously, it would be better if we could move out of the perturbative regime for $a$, as this is one assumption that severely suppresses the relative dynamics between the symmetric and antisymmetric subspaces. However, the other factor that has an even stronger effect is the variation in values $|\braket{\lambda_n}{1}|^2$. If we could instead engineer a field-free Heisenberg chain with perfect revivals and much more similar values of $|\braket{\lambda_n}{1}|^2$ (that, for example, have a smaller common denominator), that would have a far greater impact on the scaling time. Conceivably, the best that could be achieved is with $|\braket{\lambda_n}{1}|^2=\frac{1}{N}$ for all $n$, which would simply require $ak=N$. This would yield a time $T\sim N^2/\sqrt{\epsilon}$. If such a scaling could be achieved, pretty good transfer has a chance of being a relevant protocol. However, the conditions of perfect revival and $|\braket{\lambda_n}{1}|^2=\frac{1}{N}$ cannot be realised even for $N=3$ in a field-free model (a spectrum $0,-1,\sqrt{3}-2$ is required, up to scaling), and for $N=4$ there are no field-free Heisenberg models with $|\braket{\lambda_n}{1}|^2=\frac{1}{N}$ (even without the imposition of perfect revival). In the $N=4$ case, we have succeeded in creating a number of perfect revival chains, but none of them outperform the case specified in Eq.\ (\ref{eqn:PGcouplings}) for the purposes for pretty good transfer families.

As an example, we can set aside the field-free requirement\footnote{This example is purely for illustrative purposes. There are readily available perfect transfer solutions which are preferable.}, and use the solution given in \cite{kay2010a} where $h$ has off-diagonal elements
$$
J_n^2=\frac{n^2(N-n)(N+n)}{(2n-1)(2n+1)}
$$
and diagonal elements of 0. The system has a spectrum $\{-(N-1),-(N-3),\ldots,(N-3),(N-1)\}$ with a perfect revival time of $\pi$, and eigenvector elements $\braket{\lambda_n}{1}^2=\frac1N$.

\begin{figure}
\centering
\includegraphics[width=0.45\textwidth]{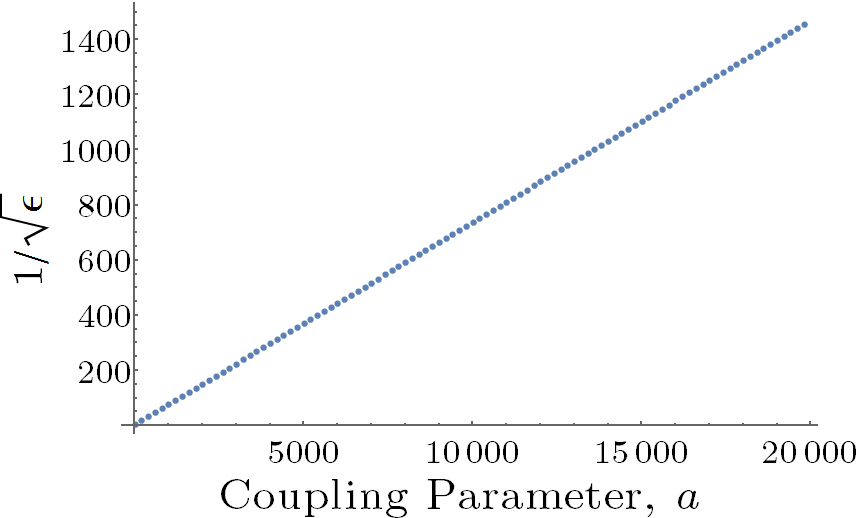}
\caption{Error $\epsilon$ in state transfer due to a mirror symmetric chain of length 40 with couplings given by the $N=20$ version of Eq.\ (\ref{eqn:lmcouple}), and a central coupling of $a$.}\label{fig:patterns}
\end{figure}

One way that we could move away from the regime of small $a$ is to take the opposite limit of $a$ being large. The antisymmetric subspace can then essentially be decomposed as a single site (the one with the field on) and a chain described by a Hamiltonian $h'$ that is the original Hamiltonian with the last row and column removed. If that $h'$ has perfect revivals at the same time, but with a relative phase of $\pi$, compared to $h$, then we achieve pretty good transfer up to the accuracy of the perturbative expansion that allows us to separate the single site. The larger $a$, the better the approximation: $\frac{1}{a}\rightarrow 0$ provides a pretty good route to state transfer. Building a chain with a prescribed spectrum for both $h$ (the symmetric subspace) and $h'$ (the perturbative expansion of the antisymmetric subspace) is a standard form of inverse eigenvalue problem. Another illustrative example, relaxing the field-free assumption, uses diagonal elements of $0$ and off-diagonal elements
\begin{equation}
J_n=\left\{\begin{array}{cc}
2\sqrt{N(N-1)}& n=N-1 \\
\sqrt{n(2N-n-1)}& \text{otherwise}
\end{array}\right.\label{eqn:lmcouple}
\end{equation}
In this case, both $h$ and $h'$ have perfect revivals at a time $\pi/2$, with a relative phase of $\pi$. The maximum eigenvalue scales as $N$, meaning that $a$ must be large compared to $N$. We have performed some numerical tests with the case $N=20$, as shown in Fig.\ (\ref{fig:patterns}). Rescaling such that all couplings and fields are bounded by a constant (e.g.\ 1), we get that the transfer time scales as $O(N/\sqrt{\epsilon})$. This is certainly optimal in terms of $N$, and indicates some of the possibilities available.

Unfortunately, it will be impossible to apply such ideas directly to field-free Heisenberg models because the proof of Theorem \ref{thm:notransfer} can be adapted to show that it is impossible. When all the $\lambda_n$ (eigenvalues of $h$) are even integers and all the $\mu_m$ (eigenvalues of $h'$ with the last row/column removed) are odd integers, $$|\braket{\lambda_1}{1}|^2=\frac{\prod_m\mu_m}{\prod_{n\neq 1}\lambda_n}.$$ The denominator clearly contains at least $N-1$ powers of 2, rendering it impossible to be $1/N$ for anything other than $N=2$. Instead, one would have to rely on a scheme in which most of the eigenvalues are chosen correctly.

\section{Conclusions}

In this paper, we have proven that there are no field-free Heisenberg chains with perfect state transfer between opposite ends of the chain. While this does not eliminate the possibility of perfect transfer between internal nodes\footnote{We have verified by brute force that perfect transfer between distinct internal nodes is impossible for $N<7$.}, it pushes one towards a consideration of both high quality transfer and pretty good state transfer. We have also argued that all solutions for high quality transfer require at least a time $O(N^2)$. This must be contrasted with cases where we introduce magnetic fields (thereby recovering solutions such as those for the exchange model \cite{christandl2004}). When we have some limited control over the system, what is the scaling of transfer time? \cite{murphy2010} used time control of magnetic fields, while we remain interested in the field-free case. This was considered in \cite{zhou2019}, but was non-committal on the scaling of the transfer time. We believe it has $O(N^2)$ scaling as any perturbative style approach must be perturbations relative to the $O(1/N^2)$ energy gaps. \cite{agundez2017} is more explicit about this, operating in a similar regime but a different numerical approach, and also achieves an $O(N^2)$ time scaling (with a large multiplicative overhead).

We have described a new paradigm for pretty good transfer, making use of the facility to tune coupling strengths, which appears to be more promising in terms of analysing the state transfer time. However, as it stands, the transfer times are prohibitive (as, we suspect, they are for all prior pretty good transfer schemes, such as \cite{coutinho2016,vanbommel2016,banchi2017}). We have outlined some future directions that can reduce the recurrence times massively under this new paradigm, towards $O(N/\sqrt{\epsilon})$, but we are yet to successfully apply them to a field-free Heisenberg model.

This study initially arose from the consideration of state synthesis questions \cite{kay2017}, in which we searched for systems with specific spectra, and had a particular null vector, such as the uniform vector (which would require a matrix of the field-free Heisenberg form). In the present setting, we also imposed centro-symmetry in order to generate the perfect state transfer. However, it is indicative that it is generally difficult to craft matrices of this form. A characterisation of which spectra are possible for a field-free Heisenberg model would be useful. Some partial steps in this direction are taken in the Appendix.

\input{no_go_theoremv3.bbl}
\newpage
\cleardoublepage
\appendix*

\section{Systems with Tunable Spectra and Ground States}

Some of the results in this paper started life with a very different purpose, which we summarise here for completeness. In addition to state transfer, there are a number of other protocols that one might be interested in accomplishing with a time-independent Hamiltonian. One of these is the preparation of interesting quantum states, particularly superpositions of a single excitation. A couple of methods have been considered for this \cite{kay2017,kay2017a}. We will focus on just one of those here \cite{kay2017}. The method starts by creating a Hamiltonian with a particular null state $\ket{\psi}=\sum_{n=1}^N\alpha_n\ket{n}$ by fixing the diagonal elements,
\begin{multline*}
h=-\sum_{n=1}^N\frac{\alpha_{n-1}J_{n-1}+\alpha_{n+1}J_n}{\alpha_n}\proj{n}+\\\sum_{n=1}^{N-1}J_{n}(\ket{n}\bra{n+1}+\ket{n+1}\bra{n}),
\end{multline*}
where we interpret $\alpha_0=\alpha_{N+1}=0$, and all others are non-zero. If we can then set all other eigenvalues to be (a multiplicative scaling of) an odd integer, then in the time $t=\pi$, an initial state $\ket{k}$ can be transformed into
$$
\ket{k}\mapsto(2\proj{\psi}-\identity)\ket{k}.
$$ 
The required $\ket{\psi}$ can hence be reverse engineered from the target state and the desired initial state.

With the $\{\alpha_n\}$ fixed, the challenge is to find the $\{J_n\}$ such that the spectral conditions are satisfied. We can now proceed just as we did in Sec.\ \ref{sec:numbers}. If we introduce
$$
E=\sum_{i=1}^{N-1}\left(\sqrt{\frac{\alpha_{i+1}}{\alpha_i}}\ket{i}-\sqrt{\frac{\alpha_i}{\alpha_{i+1}}}\ket{i+1}\right)\bra{i}
$$
with $\tilde J$ as before, then we have $h=-J_{\max}E\tilde JE^T$. For simplicity, we shall assume that all the $\alpha_n$ are positive, as are the $J_n$, such that $h$ is positive semi-definite. $h$ is similar to $\tilde h=-J_{\max}\tilde JE^TE$. Some of the spectral properties follow from \ref{lem:evalbound}.
\begin{corollary}\label{cor:main}
The eigenvalues $\lambda_n$ of $h'$ are related to the eigenvalues $\eta_n$ of $-E^TE$ by
$$
\lambda_n\geq J_{\max}\eta_n.
$$
\end{corollary}
This should help us bound, for a given target state, how long will be required to produce it via this method --- a time of at least $\pi/(J_{\max}\eta_{2})$ would be required to realise it perfectly. In general, $\eta_2$ is best calculated for any specific instance, although we can apply some general bounds. In particular, if we let
$$
\ket{\Psi}=\sum_{n=1}^{N-1}\frac{1}{\sqrt{\alpha_n\alpha_{n+1}}}\ket{n}, 
$$
then by design $E\ket{\Psi}=\frac{1}{\alpha_1}\ket{1}+\frac{1}{\alpha_N}\ket{N}$ and hence $\bra{\Psi}E^TE\ket{\Psi}=\frac{1}{\alpha^2_1}+\frac{1}{\alpha^2_N}$. Incorporating the normalisation of $\ket{\Psi}$, this proves that the smallest eigenvalue is
$$
\eta_2\geq-\frac{\frac{1}{\alpha_1^2}+\frac{1}{\alpha^2_{N}}}{\sum_{n=1}^{N-1}\frac{1}{\alpha_n\alpha_{n+1}}}.
$$
This will typically be $O(1/N)$, proving that at least a linear time is required, although this bound is often quite weak. We have already seen in the case that $\alpha_n=\frac1N$, the true smallest eigenvalue is $O(1/N^2)$.

In principle, Eqs.\ (\ref{eqn:perturb},\ref{eqn:evalperturb}) provide an iterative procedure for fixing a target spectrum via linear relations. Provided this linear problem is invertible at every step, we will rapidly converge on a good solution. The question remains when this linear problem is invertible. This is primarily governed by the number of eigenvalues one is trying to fix. It \emph{looks} like once we have fixed a particular problem instance by specifying the $\{\alpha_n\}$, there are $N-1$ coupling strengths that are free to determine the $N-1$ eigenvalues (since we know one is fixed to $0$). However, it is clear that this cannot be true in general.

\begin{lemma}
An ordered target spectrum $\{\lambda_n\}$ is impossible for $h'$ unless $\lambda_1=0$ and
$$
\frac{\lambda_k}{\lambda_N}\leq\frac{-\eta_k}{2}
$$
for all $k=2,\ldots N-1$.
\end{lemma}
\begin{proof}
By construction of $h'$, $\lambda_1=0$.

If we take a trial vector $\ket{\psi}=(\ket{i}-\ket{i+1})/\sqrt{2}$ such that $J_i=J_{\max}$, then this shows that
$$
\lambda_N\leq\bra{\psi}h'\ket{\psi}=-\frac12(4J_{\max}+J_{i+1}+J_{i-1})\leq -2J_{\max},
$$
Contrast this with $J_{\max}\geq \frac{\lambda_k}{\eta_k}$. If there is a $k$ such that
$$
\frac{\lambda_k}{\eta_k}> -\frac{\lambda_N}{2},
$$
there cannot be a satisfying value of $J_{\max}$.
\end{proof}

For example, imagine we wanted to create a Heisenberg chain, $\alpha_i=1/\sqrt{N}$, with a linear spectrum $0,-1,-2,\ldots,1-N$. With $k=2$, we are comparing $(N-1)$ with $1-\cos(\pi/N)$. Clearly the former is larger than the latter for $N>2$, and hence this choice is impossible (whether or not one tries to impose symmetry on the coupling strengths).

Even when none of these conditions is violated, such an iterative algorithm often struggles to find solutions. We are far from a complete understanding as to why.

\end{document}

%% file: no_go_theoremv3.bbl
%